\documentclass[12pt, reqno]{amsart}
\usepackage{amsmath, amsthm, amsfonts, amssymb, eucal, hyperref}

\renewcommand{\L}{\Lambda}

\newcommand{\T}{\Theta}

\newcommand{\Om}{\Omega}

\newcommand{\oq}{\ {\raise 7pt\hbox{${\scriptstyle\circ}$}}
\kern -7pt{
\hbox{$Q$}}}

\newcommand{\R}{ \mathbb R}
\newcommand{\N}{ \mathbb N}

\newcommand {\GD}{\mathfrak D}

\newcommand {\ba}{\mathbf a}

\newcommand {\bb}{\mathbf b}

\newcommand {\BD}{\mathbf D}
\newcommand {\BG}{\mathbf G}

\newcommand {\BF}{\mathbf F}

\newcommand {\BI}{\mathbf I}

\newcommand {\bx}{\mathbf x}

\newcommand {\bk}{\mathbf k}

\newcommand {\bm}{\mathbf m}
\newcommand {\bl}{\mathbf l}

\newcommand {\bn}{\mathbf n}
\newcommand {\bnu}{\boldsymbol\nu}

\newcommand {\boldeta}{\boldsymbol\eta}

\newcommand {\bxi}{\boldsymbol\xi}

\newcommand{\lu}{\langle}
\newcommand{\ru}{\rangle}

\newcommand{\CB}{\mathcal B}

\newcommand{\CA}{\mathcal A}

\newcommand{\CE}{\mathcal E}
\newcommand{\CD}{\mathcal D}

\newcommand{\1}
{{\,\vrule depth3pt height9pt}{\vrule depth3pt height9pt}
{\vrule depth3pt height9pt}{\vrule depth3pt height9pt}\,}

\DeclareMathOperator \volume {{vol}}

\newtheorem{thm}{Theorem}[section]
\newtheorem{cor}[thm]{Corollary}
\newtheorem{lem}[thm]{Lemma}
\newtheorem{prop}[thm]{Proposition}

\theoremstyle{definition}

\theoremstyle{remark}
\newtheorem{rem}[thm]{Remark}

\numberwithin{equation}{section}

\newcommand{\bee}{\begin{equation}}
\newcommand{\ene}{\end{equation}}
\newcommand{\bees}{\begin{equation*}}
\newcommand{\enes}{\end{equation*}}
\newcommand{\bes}{\begin{split}}
\newcommand{\ens}{\end{split}}

\newcommand{\bet}{\begin{thm}}
\newcommand{\ent}{\end{thm}}
\newcommand{\bel}{\begin{lem}}
\newcommand{\enl}{\end{lem}}
\newcommand{\bec}{\begin{cor}}
\newcommand{\enc}{\end{cor}}
\newcommand{\bep}{\begin{proof}}
\newcommand{\enp}{\end{proof}}
\newcommand{\ber}{\begin{rem}}
\newcommand{\enr}{\end{rem}}
\newcommand{\ep}{\varepsilon}
\newcommand{\la}{\lambda}
\newcommand{\de}{\delta}

\newcommand{\Z}{\mathbb Z}

\newcommand{\De}{\Delta}
\newcommand{\CF}{\mathcal F}

\makeatletter
\def\square{\RIfM@\bgroup\else$\bgroup\aftergroup$\fi
  \vcenter{\hrule\hbox{\vrule\@height.6em\kern.6em\vrule}\hrule}\egroup}
\makeatother

\newcommand{\Ld}{\Lambda^\dagger}
\newcommand{\Omd}{\Omega^\dagger}
\DeclareMathOperator{\rme}{\mathrm e}
\DeclareMathOperator{\rmd}{\mathrm d\!}
\newcommand{\rhs}{$\mathrm{r.\,h.\,s.}$ }

\begin{document}

\title[Lower bound on DOS
(\the\day.\the\month.\the\year)]
{Lower bound on the density of states for periodic Schr\"odinger operators}
\author[S. Morozov et al. (\the\day.\the\month.\the\year)]{Sergey Morozov \and Leonid Parnovski \and Irina Pchelintseva}
\address{Department of Mathematics\\ University College London\\
Gower Street\\ London\\ WC1E 6BT\\ UK}

\begin{abstract}
We consider Schr\"odinger operators $-\Delta+V$ in $\R^d$ ($d\geqslant 2$) with smooth
periodic potentials $V$ and prove a uniform lower bound on the density of states for large values of the spectral parameter.
\end{abstract}

\maketitle

\section{Introduction}

Let $H=-\De+V$ be a Schr\"odinger operator in $L_2(\R^d)$ with a smooth periodic potential
$V$. We will assume throughout that $d\geqslant 2$.
The {\em integrated density of states (IDS)} for $H$ is defined as
\begin{equation}\label{IDS}
N(\lambda):= \lim_{L\to \infty}L^{-d}N(\lambda; H_D^{(L)}), \quad \lambda\in\mathbb R.
\end{equation}
Here $H_D^{(L)}$ is the restriction of $H$ to the cube $[0, L]^d$ with the Dirichlet boundary conditions, and $N(\lambda; \cdot)$ is the counting function of the discrete spectrum below $\lambda$. For $H_0:= -\Delta$ the IDS can be easily computed explicitly (e.g. using the representation \eqref{densityofstates} below):
\begin{equation}\label{Laplacian IDS}
N_0(\lambda)= (2\pi)^{-d}d^{-1}\omega_d\lambda^{d/2}.
\end{equation}
Here $\omega_d= 2\pi^{d/2}/\Gamma(d/2)$ is the surface area of the unit sphere $S^{d- 1}$ in $\mathbb R^d$.

The asymptotic behaviour of the function \eqref{IDS} for large values of the spectral parameter was recently studied in a number of publications, see \cite{Karpeshina2000}, \cite{ParnovskiShterenberg2009}, and references therein.

Our article concerns the high--energy behaviour of the Radon--Ni\-ko\-dym derivative of IDS $$g:= dN/d\lambda,$$ which is called the {\em density of states (DOS)} (see \cite{ReedSimon1978}). Our main result is that for big values of $\lambda$
\begin{equation}\label{bound on g}
g(\lambda)\geqslant g_0(\lambda)\big(1- o(1)\big), 
\end{equation}
where
\begin{equation*}
g_0(\lambda)= dN_0(\lambda)/d\lambda= (2\pi)^{-d}\omega_d\lambda^{(d- 2)/2}/2.
\end{equation*}
We remark that \eqref{bound on g} should be understood in the sense of measures; in particular, we do not claim that $g(\lambda)$ is everywhere differentiable.

It has been proved in \cite{Parnovski2008} that the spectrum of $H$ contains a semi-axis
$[\la_0,+\infty)$; this statement is known as the Bethe-Sommerfeld conjecture
(see the references in \cite{Parnovski2008} for the history of this problem). This result
has an obvious reformulation in terms of IDS: each point
$\la\geqslant\la_0$ is a point of growth of $N$. It was also proved in \cite{Parnovski2008}
that for each $n\in\N$ and $\varepsilon=\la^{-n}$ we have 
\bee\label{eq:1}
N(\la+\varepsilon)- N(\la)\ll\varepsilon\lambda^{(d- 2)/2}. 
\ene
Later, when the second author discussed the results and methods of \cite{Parnovski2008}
with Yu. Karpeshina, she suggested that using the technique from that paper, one should be able to
prove the opposite bound 
\bee\label{eq:2}
N(\la+\varepsilon)- N(\la)\gg\varepsilon\lambda^{(d- 2)/2}
\ene
when $\la$ is sufficiently large,
not just with $\varepsilon=\la^{-n}$ (when the proof is relatively straightforward given 
\cite{Parnovski2008}), but also uniformly over all $\ep\in (0,1]$. In this paper we prove that for big $\lambda$
\begin{equation}\label{IDS estimate}
N(\lambda+ \varepsilon)- N(\lambda)\geqslant \frac{\omega_d}{2(2\pi)^d}\varepsilon\lambda^{(d- 2)/2}\big(1- o(1)\big).
\end{equation}
Note that \eqref{IDS estimate} implies the claimed bound \eqref{bound on g}.

The proof of \eqref{IDS estimate} is heavily based on the technique of \cite{Parnovski2008}
and uses various statements proved therein. In order to minimise the size of our paper, we will try to quote as many results as we can from
\cite{Parnovski2008}, possibly with some minor modifications when necessary.

{\bf Acknowledgement.}
As we have already mentioned, this paper is a result of observations and suggestions
made by Yu. Karpeshina; we are very grateful to her for sharing them with us and allowing
us to use them. The authors were supported by the EPSRC grant EP/F029721/1.

\section{Preliminaries}\label{preliminaries section}

We study the Schr\"odinger operator
\bee\label{Schroedinger1}
H= -\De+V(\bx), \quad \bx\in\R^d
\ene
with the potential $V$ being infinitely smooth and periodic with the lattice of periods $\L$.
We denote the lattice dual to $\L$ by $\Ld$, fundamental cells of these lattices are denoted by $\Om$ and $\Omd$, respectively. We choose $\Omd$ to be the first Brillouin zone and introduce
\begin{equation}\label{Brillouin radius}
Q:= \sup\big\{|\bxi|\big| \bxi\in\Omd\big\}.
\end{equation}
Let
\begin{equation}\label{Ds}
\BD:= -i\nabla, \quad \BD(\bk):= \BD+ \bk.
\end{equation}

The Floquet-Bloch decomposition allows to represent our operator \eqref{Schroedinger1} as a direct integral (see e.g. \cite{ReedSimon1978}):
\bee\label{directintegral}
H=\int_{\Omd}\oplus H(\bk)\rmd\bk,
\ene
where
\begin{equation}\label{skewed operators}
H(\bk)=\BD(\bk)^2 +V(\bx)
\end{equation}
is the family of `fibre' operators acting in
$L_2(\Om)$.
The domain of each $H(\bk)$ is the set of periodic functions from $H^2(\Om)$.
The spectrum of $H$ is the union over $\bk\in\Omd$ of the spectra of the operators \eqref{skewed operators}.

We denote by $|\cdot|_\circ$ the surface area Lebesgue measure on the unit sphere $S^{d- 1}$ in $\mathbb R^d$ and put $\omega_d:= |S^{d- 1}|_\circ= 2\pi^{d/2}/\Gamma(d/2)$.
Finally,
\bee\label{densityofstates}
N(\la):= (2\pi)^{-d}\int_{\Omd}\# \big\{j:\,\la_j(\bk)<\la\big\}\rmd\bk
\ene
is the {\it integrated density of states} of the operator \eqref{Schroedinger1}. It is known (see e.g. \cite{ReedSimon1978}) that the definitions \eqref{IDS} and \eqref{densityofstates} are equivalent.

The main result of the paper is

\begin{thm}\label{maintheorem2}
For sufficiently big $\lambda$ and any $\varepsilon> 0$ the integrated density of states of $H$ satisfies \eqref{IDS estimate}.
\end{thm}

By $B(R)$ we denote the ball of radius $R$ centered at the origin.
Given two positive functions $f$ and $g$,
we say that $f\gg g$, or $g\ll
f$, or $g=O(f)$ if the ratio $g/f$ is bounded. We say
$f\asymp g$ if $f\gg g$ and $f\ll g$.
Whenever we use $O$, $o$, $\gg$, $\ll$, or $\asymp$ notation, the
constants involved can depend on $d$ and norms of the potential in various Sobolev spaces $H^s$; the same is also the case when we use the
expression `sufficiently large'.
By $\la=\rho^2$ we denote a point on the spectral axis.
We also denote by $v$ the $L_{\infty}$--norm of the potential $V$, and put
$J:=[\la-20v,\la+20v]$.
Let
\bee\label{Arho1intro}
\CA:=\Big\{\bxi\in{\mathbb R}^d,\, \,\big| |\bxi|^2-\la\big|\leqslant 40 v \Big\}.
\ene
Notice that the definition of $\CA$ obviously implies that if $\bxi\in\CA$, then
$\big| |\bxi|-\rho\big|\ll\rho^{-1}$.

Any vector $\bxi\in\R^d$ can be
uniquely decomposed as $\bxi=\bn+\bk$ with $\bn\in\Ld$ and
$\bk\in \Omd$. We call $\bn=[\bxi]$ the `integer part' of $\bxi$ and
$\bk=\{\bxi\}$ the `fractional part' of $\bxi$.

By $\volume(\cdot)$ we denote the Lebesgue measure in $\mathbb R^d$.
The identity matrix is denoted by $\BI$.
For any $h\in L_2(\Om)$ we introduce its Fourier coefficients
\begin{equation}\label{Fourier transform}
h_{\mathbf n}:= (\volume\Om)^{-1/2}\int_{\Om}h(\mathbf x)\exp\big(-i \langle\mathbf n, \mathbf x\rangle\big)\rmd\mathbf x, \quad \mathbf n\in \Ld.
\end{equation}

For $\bxi\in \R^d\setminus\{\mathbf 0\}$ we define $r= r(\bxi):= |\bxi|$ and $\bxi':= \bxi/|\bxi|$. We put
\begin{equation}\label{R}
R= R(\rho):= \rho^{1/(36d^2(d+2))}
\end{equation}
(so that the condition stated after equation (5.15) in \cite{Parnovski2008} is satisfied). For $j\in \N$ let
\begin{equation*}
\T_j':= \Ld\cap B(jR)\setminus\{\mathbf 0\}.
\end{equation*}
Let $M:= 5d^2+ 7d$.
We introduce the set
\bee\label{B}
\CB:= \Big\{\bxi\in\CA\Big| \big|\langle\bxi, \boldeta'\rangle\big|>\rho^{1/2},\;\textrm{for all}\; \boldeta\in\T'_{6M}\Big\}.
\ene
In other words, $\CB$ consists of all points $\bxi\in\CA$ the projections of
which to the directions of all vectors $\boldeta\in\T'_{6M}$ have lengths larger than $\rho^{1/2}$.
We also denote $\CD:=\CA\setminus\CB$.

In the rest of the section we quote some results from \cite{Parnovski2008} which we will use in this paper. Our approach is slightly different from that of \cite{Parnovski2008}. In particular, we consider arbitrary lattice
of periods $\Lambda$, not equal to $(2\pi\Z)^d$. We also use a different form of the  Floquet-Bloch decomposition (so that the operators on fibers \eqref{skewed operators} are defined on the same domain). This leads to several straightforward changes in the formulation of the results from \cite{Parnovski2008}.
These changes are:

\begin{enumerate}
\item The lattices $(2\pi\Z)^d$ and $\Z^d$ are replaced by $\Lambda$ and $\Ld$, respectively. The `integer' and `fractional' parts are now defined with respect to $\Ld$ (see above);
\item The matrices $\BF$ and $\BG$ are replaced by the unit matrix $\BI$ throughout;
\item The Fourier transform is now defined by \eqref{Fourier transform}, and the exponentials $e_\bm$ introduced at the beginning of Section~5 in \cite{Parnovski2008} are redefined as
\begin{equation*}
e_{\bm}(\bx):= (\volume\Om)^{-1/2} \rme^{i\lu\bm,\bx\ru},\ \ \bm \in\Ld;
\end{equation*}
\item The operators $H(\bk)$ are now given by \eqref{skewed operators} on the common domain $\GD$.
\end{enumerate}

The main result we will need follows from Corollary~7.15 of \cite{Parnovski2008}:
\begin{prop}\label{maincorollary1}
There exist mappings $f,g:\CA\to\R$ which
satisfy the following properties:

(i) $f(\bxi)$ is an eigenvalue of $H(\bk)$ with $\{\bxi\}=\bk$;
$\big|f(\bxi)-|\bxi|^2\big|\leqslant 2v$. $f$ is an injection (if we count all eigenvalues with multiplicities)
and all eigenvalues of $H(\bk)$ inside
$J$ are in the image of $f$.

(ii) If $\bxi\in\CA$, then $\big|f(\bxi)-g(\bxi)\big|\leqslant \rho^{-d- 3}$.

(iii) For any $\bxi\in\CB$
\bee\label{eq1:maincorollary1}
\bes
&g(\bxi)=|\bxi|^2\\
&+\sum_{j=1}^{2M} \sum_{\boldeta_1,\dots,\boldeta_j\in \T'_M}
\sum_{2\leqslant n_1+\dots+n_j\leqslant 2M}C_{n_1,\dots,n_j}
\lu\bxi,\boldeta_1\ru^{-n_1}\dots\lu\bxi, \boldeta_{j}\ru^{-n_j}.
\end{split}
\ene

(iv) Let $I=[\ba,\bb]\subset\CA$ be a straight interval
of length $L:=|\bb-\ba|\ll\rho^{-1}$.
Then there exists an integer vector $\bn$ such that
$\big|g(\bb+\bn)-g(\ba)\big|\ll L\rho+\rho^{-d- 3}$. Moreover, suppose $\bm\ne 0$
is an integer vector such that the interval $I+\bm$ is also entirely inside $\CA$.
Then there exist two different integer vectors $\bn_1$ and $\bn_2$ such that
$\big|g(\bb+\bn_1)-g(\ba)\big|\ll L\rho+\rho^{-d- 3}$ and
$\big|g(\bb+\bn_2)-g(\ba+\bm)\big|\ll L\rho+\rho^{-d- 3}$.
\end{prop}

\ber
Formula \eqref{eq1:maincorollary1} implies that
\bee\label{partialg}
\partial g/\partial r(\bxi)\asymp \rho, \quad \textrm{for any}\quad \bxi\in \CB.
\ene
\enr

For each positive $\de\leqslant v$ we denote by $\CA(\de)$, $\CB(\de)$, and $\CD(\de)$ the intersections
of $g^{-1}\big([\rho^2-\de,\rho^2+\de]\big)$ with
$\CA$, $\CB$, and $\CD$, respectively.

It is proved in Lemma~8.1 of \cite{Parnovski2008} that
\bee\label{volumeCD}
\volume\big(\CD(\de)\big)\ll \rho^{d-7/3}\de.
\ene

The following statement (Corollary~8.5 of \cite{Parnovski2008}) gives a sufficient condition for the continuity of $f$:

\bel\label{cor:continuousf} There is a constant $C_1$ with the following properties.
Let
\bees
I:=\big\{\bxi(t):\, t\in[t_{min},t_{max}]\big\}\subset \CB(v).
\enes
be a straight interval of length $L<\rho^{-1}\de$.
Suppose that there is a
point $t_0\in[t_{min},t_{max}]$ with the property that for each non-zero
$\bn\in\Ld$ $g\big(\bxi(t_0)+\bn\big)$ is either outside the interval
\bees
\Big[g\big(\bxi(t_0)\big)-C_1\rho^{-d- 3}-C_1\rho L,\; g\big(\bxi(t_0)\big)+C_1\rho^{-d- 3}+C_1\rho L\Big]
\enes
or not defined. Then $f\big(\bxi(t)\big)$ is a continuous function of $t$.
\enl

By inspection of the proof of Lemma~8.3 of \cite{Parnovski2008} we obtain

\bel\label{volumeCB1}
For large enough $\rho$ and $\de<\rho^{-1}$ the following estimates hold uniformly over $\ba\in\Ld\setminus\{\mathbf 0\}$:
if $d\geqslant 3$,
\bee\label{eq:volumeCB1}
\volume\Big(\CB(\de)\cap \big(\CB(\de)+\ba\big)\Big)\ll (\de^2\rho^{d-3}+\de\rho^{-d});
\ene
if $d= 2$,
\begin{equation}\label{intersection volume}
\volume\Big(\mathcal B(\delta)\cap\big(\mathcal B(\delta)+ \mathbf a\big)\Big)\begin{cases}\ll \delta^{3/2}, & |\mathbf{a}|\leqslant 2\rho- 1,\\ \ll \delta^{3/2}+ \delta\rho^{-2},& \big| |\mathbf{a}|- 2\rho\big|< 1,\\ = 0,& |\mathbf{a}|\geqslant 2\rho+ 1.\end{cases}
\end{equation}
\enl

\section{Prevalence of regular directions}\label{good angles section}

\begin{lem}\label{good angles lemma}
For $\rho$ big enough and
\begin{equation*}\label{delta bound}
0< \de\leqslant \rho^{-d- 3}
\end{equation*}
there exists a set $\CF= \CF(\rho)$ on the unit sphere $S^{d- 1}$ in $\mathbb R^d$ with
\begin{equation}\label{Psi condition}
|\CF|_\circ\geqslant \omega_d\big(1- o(1)\big)
\end{equation}
such that $f(\boldsymbol\xi)$ is a simple eigenvalue of $H\big(\{\boldsymbol\xi\}\big)$ continuously depending on $r:= |\bxi|$ for every $\boldsymbol\xi= (r, \boldsymbol\xi')\in f^{-1}\big([\rho^2- \de, \rho^2+ \de)\big)$ with $\boldsymbol\xi':= \bxi/|\bxi|\in \CF$.
\end{lem}

\bep
It is enough to consider $\delta:= \rho^{-d- 3}$. For each $\bxi'\in S^{d- 1}$ let
\begin{equation}\label{I_eta set}
I_{\boldsymbol\xi'}(\delta):= \{r\boldsymbol\xi', r> 0\}\cap \mathcal B(\delta).
\end{equation}
Let $\CF_1:= \big\{\boldsymbol\xi'\in S^{d- 1}\big| I_{\boldsymbol\xi'}(\delta)\neq\emptyset,\; \overline{I_{\boldsymbol\xi'}(\delta)}\cap\mathcal D(\delta)=\emptyset\big\}$.

For any $\boldeta\in\T'_{6M}$ the area of the set of points $\bxi'\in S^{d- 1}$ satisfying
\begin{equation*}
\big|\langle r\bxi',\boldeta'\rangle\big|\leqslant \rho^{1/2}
\end{equation*}
is evidently $O(\rho^{- 1/2})$ if $r\geqslant \rho/2$ (the latter is true for all $r\bxi'\in \CA$). Since the number of elements
in $\T'_{6M}$ is $O(R^d)$, by \eqref{R} and \eqref{B} we have
\begin{equation}\label{Psi_2 to S^d-1}
|S^{d- 1}\setminus \CF_1|_\circ= o(1).
\end{equation}

By definition $\mathcal B(\delta)= \mathcal B\cap g^{-1}\big([\rho^2- \delta, \rho^2+ \delta]\big)$, hence \eqref{partialg} implies that for big $\rho$ the length $l_{\boldsymbol\xi'}(\delta)$ of $I_{\boldsymbol\xi'}(\delta)$ satisfies
\begin{equation}\label{length estimate 1}
l_{\boldsymbol\xi'}(\delta)\asymp \delta\rho^{-1}, \quad \boldsymbol\xi'\in \CF_1.
\end{equation}

Let
\begin{equation*}
\CF:= \big\{\boldsymbol\xi'\in \CF_1\big\arrowvert f \textrm{ is continuous on }I_{\boldsymbol\xi'}(\delta)\big\},
\end{equation*}
and
\begin{equation*}
\mathcal E(\de):= \big\{\bxi\in\CB(\de)\big| \bxi'\in \CF_1\setminus\CF\big\}.
\end{equation*}
Lemma~\ref{cor:continuousf} tells us
that for each point $\bxi\in\CE(\de)$ there is a non-zero vector $\bn\in\Ld$ such that
\bee
\big|g(\bxi+\bn)-g(\bxi)\big|\leqslant C_1\big(\rho^{-d- 3}+\rho l_{\boldsymbol\xi'}(\de)\big)\ll(\rho^{-d- 3}+\de).
\ene
Since $\big|g(\bxi)-\rho^2\big|\leqslant\de$, this implies
\begin{equation*}
\big|g(\bxi+\bn)-\rho^2\big|\leqslant C_2(\rho^{-d- 3}+\de)=:\de_1\ll \rho^{-d- 3}= \delta,
\end{equation*}
and thus $\bxi+\bn\in\CA(\de_1)$; notice that $C_2> 1$ and so $\de_1> \de$.
Therefore, each point $\bxi\in\CE(\de)$ also belongs to the set $\big(\CA(\de_1)-\bn\big)$
for a non-zero $\bn\in\Ld$; obviously,
$|\bn|\ll\rho$. In other words,
\bee\label{cover1}
\CE(\de)\subset \bigcup_{\bn\in\Ld\cap B(C\rho), \,\bn\ne 0}
\bigl(\CA(\de_1)-\bn\bigr)=\bigcup_{\bn\ne 0}\bigl(\CB(\de_1)-\bn\bigr)\cup
\bigcup_{\bn\ne 0}\bigl(\CD(\de_1)-\bn\bigr).
\ene
To proceed further, we need more notation. Denote $\CD_0(\de_1)$ to be the set
of all points $\bnu$ from $\CD(\de_1)$ for which there is no non-zero $\bn\in\Ld$
satisfying $\bnu-\bn\in\CB(\de)$;  $\CD_1(\de_1)$ to be the set
of all points $\bnu$ from $\CD(\de_1)$ for which there is a unique non-zero $\bn\in\Ld$
satisfying $\bnu-\bn\in\CB(\de)$; and $\CD_2(\de_1)$ to be the rest of the points from $\CD(\de_1)$ (i.e. $\CD_2(\de_1)$ consists of all points $\bnu$ from $\CD(\de_1)$ for which there exist at least two different non-zero vectors $\bn_1,\bn_2\in\Ld$ satisfying $\bnu-\bn_j\in\CB(\de)$). Then Lemma~8.7 of \cite{Parnovski2008} implies that we can rewrite \eqref{cover1} as
\bee\label{cover2}
\CE(\de)
\subset\bigcup_{\bn\ne \mathbf0}\bigl(\CB(\de_1)-\bn\bigr)\cup
\bigcup_{\bn\ne \mathbf0}\bigl(\CD_1(\de_1)-\bn\bigr).
\ene
This, obviously, implies
\bee\label{cover3}
\CE(\de)\subset \bigcup_{\bn\ne \mathbf0}\Big(\big(\CB(\de_1)-\bn\big)\cap\CB(\de)\Big)\cup
\bigcup_{\bn\ne \mathbf0}\Big(\big(\CD_1(\de_1)-\bn\big)\cap\CB(\de)\Big),
\ene
since $\CE(\de)\subset \CB(\de)$.

The definition of the set $\CD_1(\de_1)$ and \eqref{volumeCD} imply that
\bee\label{volumecover1}
\bes
&\volume\bigg(\bigcup_{\bn\neq\mathbf0}\Bigl(\bigl(\CD_1(\de_1)-\bn\bigr)
\cap\CB(\de)\Bigr)\bigg)\leqslant \volume\big(\CD_1(\de_1)\big)\\& \leqslant \volume\big(\CD(\de_1)\big)
\ll  \de_1\rho^{d-7/3}\ll \de\rho^{d- 7/3}.
\end{split}
\ene
For $d\geqslant 3$ Lemma~\ref{volumeCB1}, inequality $\de< \de_1$, and the fact that the union in \eqref{cover3} consists of no more than $C\rho^d$ terms imply
\bee\label{volumecover2}
\volume\bigg(\bigcup_{\bn\neq\mathbf0}\Bigl(\bigl(\CB(\de_1)-\bn\bigr)
\cap\CB(\de)\Bigr)\bigg)\ll\rho^d(\de_1^2\rho^{d-3}+\de_1\rho^{-d})\ll \de(\rho^{d- 6}+ 1).
\ene
For $d= 2$ we obtain by Lemma~\ref{volumeCB1}
\begin{equation}\label{new (8.25)}\begin{split}
&\volume\bigg(\underset{\mathbf n\in \Ld\setminus\{\mathbf0\}}\bigcup \Big(\mathcal B(\delta)\cap\big(\mathcal B(\delta_1)+ \mathbf n\big)\Big)\bigg)\\ &\leqslant \sum_{\substack{\mathbf n\in \Ld\setminus\{\mathbf0\}\\ |\mathbf{n}|\leqslant 2\rho- 1}}\volume\Big(\mathcal B(\delta)\cap\big(\mathcal B(\delta_1)+ \mathbf n\big)\Big)\\ &+ \sum_{\substack{\mathbf n\in \Ld\setminus\{\mathbf0\}\\ ||\mathbf{n}|- 2\rho|< 1}}\volume\Big(\mathcal B(\delta)\cap\big(\mathcal B(\delta_1)+ \mathbf n\big)\Big)\\ &\ll \delta_1^{3/2}\rho^2+ \rho(\delta_1^{3/2}+ \delta_1\rho^{-2})\ll \de\rho^{-1/2},
\end{split}\end{equation}
where we have used that 
\begin{equation*}
\#\Big\{\mathbf n\in \Ld\Big|\big||\mathbf{n}|- 2\rho\big|< 1\Big\}\ll \rho.
\end{equation*}

Applying \eqref{volumecover1}, \eqref{volumecover2}, and \eqref{new (8.25)} to \eqref{cover3} we obtain for all $d\geqslant 2$
\begin{equation}\label{CE estimate}
\volume \CE(\de)\ll \de\rho^{d- 7/3}.
\end{equation}
By definition,
\begin{equation*}
\mathcal E(\delta)= \bigcup_{\boldsymbol\xi'\in\CF_1\setminus\CF}I_{\boldsymbol\xi'}(\delta).
\end{equation*}
Hence by \eqref{length estimate 1}
\begin{equation}\label{area of Psi_2 without Psi}
|\CF_1\setminus\CF|_\circ\ll \delta^{-1}\rho^{2- d}\volume\mathcal E(\delta).
\end{equation}
Combining \eqref{CE estimate} and \eqref{area of Psi_2 without Psi} we conclude that for big $\rho$
\begin{equation}\label{good bound on Psi_2 without Psi}
|\CF_1\setminus\CF|_\circ= o(1).
\end{equation}
We have
\begin{equation}\label{complement area}
|S^{d- 1}\setminus \CF|_\circ= |S^{d- 1}\setminus \CF_1|_\circ+ |\CF_1\setminus\CF|_\circ.
\end{equation}
Substituting \eqref{Psi_2 to S^d-1} and \eqref{good bound on Psi_2 without Psi} into \eqref{complement area} we obtain \eqref{Psi condition}.

Now we notice that for every $\bxi'\in \CF$ the interval $I_{\bxi'}(\de)$ has the following property: for each point $\bxi\in I_{\bxi'}(\de)$ and
each non-zero vector $\bn\in\Ld$ such that $\bxi+\bn\in\CA$ we have
$\big|g(\bxi+\bn)-g(\bxi)\big|>2\rho^{-d- 3}$. This implies $f(\bxi+\bn)-f(\bxi)\ne 0$. Therefore,
$f(\bxi)$ is a simple eigenvalue of $H\big(\{\bxi\}\big)$ for each $\bxi\in I_{\bxi'}(\de)$.
The lemma is proved.
\enp

\section{Some properties of operators on the fibers}

For $m\in\R$ let
\begin{equation*}
V^{(m)}:= \Big(\sum_{\mathbf n\in \Ld}|\mathbf n|^{2m}|V_{\mathbf n}|^2\Big)^{1/2}.
\end{equation*}
Since $V$ is smooth, $V^{(m)}$ is finite for any $m\geqslant 0$. Recall that $Q$ is defined by \eqref{Brillouin radius}.

\begin{lem}\label{l1}
Fix $m\in \mathbb N$ and $\varkappa\in (0, 1)$. For $\mathbf k\in \Omd$ let $\psi$ be a normalized eigenfunction of $H(\mathbf k)$:
\begin{equation}\label{g16}
H(\mathbf k)\psi = \zeta\psi
\end{equation}
with the eigenvalue
\begin{equation}\label{zeta condition}
\zeta\geqslant \max\big\{36Q^2\varkappa^{-2}, (1+ m\varkappa)^{2/(d- 1)}\varkappa^{-2d/(d- 1)}\big\}.
\end{equation}
Then there exists $M_m= M_m(d, \Lambda, V)\in \R_+$ such that for all $\mathbf n\in\Ld$ with
\begin{equation}\label{n assumption}
|\mathbf n|\geqslant (1+ m\varkappa)\sqrt\zeta
\end{equation}
the Fourier coefficients of $\psi$ satisfy
\begin{equation}\label{coefficient decay}
|\psi_\mathbf n|< M_m\varkappa^{-m}|\mathbf n|^{-(3m+ 1)/2}.
\end{equation}
\end{lem}

\begin{proof}{}
We proceed by induction. Suppose that either $m= 1$, or $m> 1$ and the statement is proved for $m- 1$.
Substituting the Fourier series
\begin{equation*}\label{star}
\psi(\bx) = (\volume \Om)^{-1/2}\sum_{\mathbf n\in \Ld}\psi_\mathbf n\exp\big(i\langle\mathbf n, \mathbf x\rangle\big), \quad \bx\in\Om
\end{equation*}
into \eqref{g16} and equating the coefficients at $\exp\big(i\langle\mathbf n, \mathbf x\rangle\big)$ on both sides, we obtain by \eqref{skewed operators}:
\begin{equation}\label{eq4}
|\mathbf n+\mathbf k|^2 \psi_\mathbf n+ \sum\limits_{\mathbf l\in \Ld} V_{\mathbf n-\mathbf l}\psi_\mathbf l = \zeta \psi_\mathbf n.
\end{equation}
Since $|\mathbf k|\leqslant Q$, by \eqref{zeta condition} and \eqref{n assumption} we have
\begin{equation}\label{nk estimate}
2|\bn||\bk|\leqslant \varkappa|\bn|^2/6+ 6\varkappa^{-1}Q^2\leqslant \varkappa|\bn|^2/3.
\end{equation}
For $\varkappa\in(0, 1)$, it follows from \eqref{n assumption} that
\begin{equation}\label{n-zeta estimate}
|\bn|^2- \zeta\geqslant \big(1- (1+ \varkappa)^{-2}\big)|\bn|^2= \varkappa(2+ \varkappa)(1+ \varkappa)^{-2}|\bn|^2\geqslant \varkappa|\bn|^2/2.
\end{equation}
Combining \eqref{nk estimate} and \eqref{n-zeta estimate} we obtain
\begin{equation*}
|\bn+ \bk|^2- \zeta\geqslant |\bn|^2- 2|\bn||\bk|- \zeta\geqslant \varkappa|\bn|^2/6,
\end{equation*}
and thus by \eqref{eq4}
\begin{equation}\label{halfway}
|\psi_\mathbf n|< 6\varkappa^{-1}|\mathbf n|^{-2}\sum\limits_{\mathbf l\in \Ld}|V_{\mathbf n-\mathbf l}\psi_\mathbf l|.
\end{equation}

If $m= 1$ we estimate the sum on the \rhs by $V^{(0)}$ using Cauchy--Schwarz inequality (since $\psi$ is normalized) and obtain \eqref{coefficient decay} with $M_1:= 6V^{(0)}$.

If $m> 1$, we estimate
\begin{equation}\label{inner stuff}
\sum\limits_{\bl\in \Ld:\, |\mathbf l- \mathbf n|\leqslant |\mathbf n|^{1/d}}|V_{\mathbf n-\mathbf l}\psi_{\mathbf l}|\leqslant \underset{\bm: \,|\mathbf m|\geqslant |\mathbf n|- |\mathbf n|^{1/d}}\sup|\psi_\mathbf m|\sum_{\bl\in \Ld:\, |\mathbf l|\leqslant |\mathbf n|^{1/d}}|V_\mathbf l|.
\end{equation}
By \eqref{n assumption}, \eqref{zeta condition}, and monotonicity of the function $q(t)= t- t^{1/d}$ for $t> 1$ we have
\begin{equation*}
|\mathbf n|- |\mathbf n|^{1/d}\geqslant (1+ m\varkappa)\sqrt\zeta- \big((1+ m\varkappa)\sqrt\zeta\big)^{1/d}\geqslant \big(1+ (m- 1)\varkappa\big)\sqrt\zeta.
\end{equation*}
We can thus apply the induction hypothesis obtaining
\begin{equation}\label{induction}
\underset{\bm:\, |\mathbf m|\geqslant |\mathbf n|- |\mathbf n|^{1/d}}\sup|\psi_\mathbf m|\leqslant \varkappa^{1- m}M_{m- 1}\big(1- |\mathbf n|^{(1- d)/d}\big)^{1- 3m/2}|\mathbf n|^{1- 3m/2}.
\end{equation}
Since $\varkappa\in (0, 1)$, by \eqref{n assumption} and \eqref{zeta condition} we have
\begin{equation*}
|\bn|\geqslant (1+ m\varkappa)\sqrt\zeta\geqslant (\varkappa^{-1}+ m)^{d/(d- 1)}> 2^{d/(d- 1)},
\end{equation*}
thus
\begin{equation}\label{>2^m}
\big(1- |\bn|^{(1- d)/d}\big)^{1- 3m/2}< 2^{3m/2- 1}.
\end{equation}

Let
\begin{equation*}
W:= \sup_{r> 1}r^{-d}\#\big\{\mathbf l\in\Ld\big| |\mathbf l|\leqslant r\big\}.
\end{equation*}
Clearly, $W< \infty$. Then by Cauchy--Schwarz inequality
\begin{equation}\label{V in ball}
\sum_{\bl\in \Ld:\, |\mathbf l|\leqslant |\mathbf n|^{1/d}}|V_\mathbf l|\leqslant W^{1/2}V^{(0)}|\mathbf n|^{1/2}.
\end{equation}
Substituting \eqref{induction}, \eqref{>2^m},  and \eqref{V in ball} into \eqref{inner stuff} we get
\begin{equation}\label{better inner stuff}
\sum\limits_{\bl\in \Ld:\, |\mathbf l- \mathbf n|\leqslant |\mathbf n|^{1/d}}|V_{\mathbf n-\mathbf l}\psi_{\mathbf l}|< 2^{3m/2- 1}\varkappa^{1- m}W^{1/2}V^{(0)}M_{m- 1}|\mathbf n|^{3(1- m)/2}.
\end{equation}
On the other hand, since $\|\psi\|= 1$, by Cauchy--Schwarz we have
\begin{equation}\label{outer stuff}\begin{split}
&\sum\limits_{\bl\in \Ld:\, |\mathbf l- \mathbf n|> |\mathbf n|^{1/d}}|V_{\mathbf n-\mathbf l}\psi_{\mathbf l}|< |\mathbf n|^{3(1- m)/2}\sum\limits_{\bl\in \Ld:\, |\mathbf l|> |\mathbf n|^{1/d}}|\mathbf l|^{3(m- 1)d/2}|V_{\mathbf l}||\psi_{\bn- \bl}|\\ &\leqslant |\mathbf n|^{3(1- m)/2}\Big(\sum\limits_{\bl\in \Ld:\, |\mathbf l|> |\mathbf n|^{1/d}}|\mathbf l|^{3(m- 1)d}|V_{\mathbf l}|^2\Big)^{1/2}\leqslant V^{(3(m- 1)d/2)}|\mathbf n|^{3(1- m)/2}.
\end{split}\end{equation}
Inserting \eqref{better inner stuff} and \eqref{outer stuff} into \eqref{halfway} we obtain \eqref{coefficient decay} with
\begin{equation*}
M_m:= 6(2^{3m/2- 1}W^{1/2}V^{(0)}M_{m- 1}+ V^{(3(m- 1)d/2)}).
\end{equation*}
\end{proof}

\begin{lem}\label{gradient bound lemma}
For any $\eta\in (0, 1)$ there exists $\zeta_0> 0$ such that if $\zeta(\mathbf k)\geqslant \zeta_0$ is a simple eigenvalue of $H(\mathbf k)$ for some $\bk\in\Omd$ then
\begin{equation}\label{gradient bound}
|\nabla_\bk\zeta|\leqslant 2(1+ \eta)\sqrt\zeta.
\end{equation}
\end{lem}

\begin{proof}{}
Let $\psi(\mathbf k)$ be the eigenfunction corresponding to $\zeta(\bk)$ with
\begin{equation}\label{Psi normalization}
\big\|\psi(\mathbf k)\big\|= 1.
\end{equation}
Then
\begin{equation}\label{1g}
\nabla_{\mathbf k}\zeta(\mathbf k)= \nabla_{\mathbf k}\big(\psi(\mathbf k), H(\mathbf k)\psi(\mathbf k)\big)= \Big(\psi(\mathbf k), \big(\nabla_{\mathbf k}H(\mathbf k)\big)\psi(\mathbf k)\Big).
\end{equation}
By \eqref{skewed operators} and \eqref{Ds},
\begin{equation*}
\nabla_{\mathbf k}H(\mathbf k)= 2\mathbf D(\mathbf k).
\end{equation*}
Substituting this into (\ref{1g}) we obtain:
\begin{equation}\label{2g}
\big|\nabla_{\mathbf k}\zeta(\mathbf k)\big|\leqslant 2\big\|\mathbf D(\mathbf k)\psi(\mathbf k)\big\|= 2\Big(\sum\limits_{\mathbf n\in \Ld}|\mathbf n+ \mathbf k|^2\big|\psi_\mathbf n(\mathbf k)\big|^2\Big)^{1/2}.
\end{equation}
Let
\begin{equation}\label{m choice}
m:= \big[(d+ 1)/3\big]+ 1
\end{equation}
and
\begin{equation}\label{kappa choice}
\varkappa:= \eta/(2m+ 1).
\end{equation}
We put
\begin{equation}\label{eta choice}
\zeta_0:= \max\big\{36Q^2\varkappa^{-2}, (1+ m\varkappa)^{2/(d- 1)}\varkappa^{-2d/(d- 1)}\big\}
\end{equation}
and assume that $\zeta:= \zeta(\mathbf k)\geqslant \zeta_0$.
Since by \eqref{Brillouin radius} $|\mathbf k|\leqslant Q$, by \eqref{Psi normalization}, \eqref{eta choice}, and \eqref{kappa choice} we have
\begin{equation}\label{inner sum}
\sum\limits_{|\mathbf n|< (1+ m\varkappa)\sqrt\zeta}|\mathbf n+ \mathbf k|^2\big|\psi_\mathbf n(\mathbf k)\big|^2< \big(1+ (m+ 1/6)\varkappa\big)^2\zeta< (1+ \eta/2)^2\zeta.
\end{equation}
For $|\mathbf n|\geqslant (1+ m\varkappa)\sqrt\zeta$ we apply Lemma~\ref{l1} obtaining
\begin{equation}\label{outer sum}
\sum\limits_{|\mathbf n|\geqslant (1+ m\varkappa)\sqrt\zeta}|\mathbf n+ \mathbf k|^2\big|\psi_\mathbf n(\mathbf k)\big|^2\leqslant M_m^2\varkappa^{-2m}\sum\limits_{|\mathbf n|\geqslant (1+ m\varkappa)\sqrt\zeta}|\mathbf n+ \mathbf k|^2\big|\mathbf n|^{-3m- 1}.
\end{equation}
By \eqref{m choice} the \rhs of \eqref{outer sum} is finite and is $O(\zeta^{-1/2})$. Thus, increasing $\zeta_0$ if necessary, by \eqref{2g}, \eqref{inner sum}, and \eqref{outer sum} we obtain \eqref{gradient bound}.
\end{proof}

\section{Proof of Theorem~\ref{maintheorem2}}\label{BS section}

It is enough to prove 

\bet\label{short intervals theorem}
For any $\alpha\in (0, 1)$ there exists $\rho_0> 0$ big enough such that for all $\rho\geqslant \rho_0$
\begin{equation}\label{small interval statement}
N(\rho^2+ \de)- N(\rho^2- \de)\geqslant (1- \alpha)(2\pi)^{-d}\omega_d\de\rho^{d- 2}
\end{equation}
for any
\begin{equation}\label{bound on epsilon}
0 <\de\leqslant \rho^{-d- 3}.
\end{equation}
\ent

Indeed, the original statement of Theorem~\ref{maintheorem2} can be obtained by partitioning of the interval $[\lambda, \lambda+ \varepsilon]$ into subintervals with lengths not exceeding $2\lambda^{(-d- 3)/2}$ and adding up estimates \eqref{small interval statement} on this subintervals (with $\rho^2$ being respective middle points).

\bep
We first express the growth of IDS in terms of the function $f$ of Proposition~\ref{maincorollary1}:
\begin{equation}\label{volume formula}
N(\rho^2+ \de)- N(\rho^2- \de)= (2\pi)^{-d}\volume\big(f^{-1}[\rho^2- \de, \rho^2+ \de)\big).
\end{equation}
We can write
\begin{equation}\label{angular decomposition}
\volume\big(f^{-1}[\rho^2- \de, \rho^2+ \de)\big)= \int_{S^{d- 1}}\int_0^\infty \chi(r, \boldsymbol\xi')r^{d- 1}\rmd r \rmd\boldsymbol\xi',
\end{equation}
where $\chi$ is the indicator function of $f^{-1}\big([\rho^2- \de, \rho^2+ \de)\big)$.
To obtain a lower bound we can restrict the integration in \eqref{angular decomposition} to $\bxi'\in \CF$ defined in Lemma~\ref{good angles lemma}. Then for any $\eta\in (0, 1)$ there exists $\rho_0> 0$ such that for any $\rho\geqslant \rho_0$ we have
\begin{equation}\label{Psi condition implemented}
|\CF|_\circ\geqslant (1- \eta)\omega_d,
\end{equation}
and for any $\bxi'\in \CF$ the support of $\chi(\cdot, \boldsymbol\xi')$ contains an interval $[r_1, r_2]$ with
\begin{equation}\label{all good things}
(1- \eta)\rho\leqslant r_1< r_1+ (1- \eta)\rho^{-1}\delta\leqslant r_2.
\end{equation}
Indeed, the first inequality in \eqref{all good things} follows from Proposition~\ref{maincorollary1}(ii),(iii). The last inequality in \eqref{all good things} follows from Lemmata~\ref{good angles lemma} and \ref{gradient bound lemma}.

Thus for all $\rho\geqslant \rho_0$ by \eqref{Psi condition implemented} and \eqref{all good things} we obtain
\begin{equation}\label{volume calculation}\begin{split}
\int_{S^{d- 1}}\int_0^\infty \chi(r, \boldsymbol\xi')r^{d- 1}\rmd r \rmd\boldsymbol\xi'&\geqslant \int_\CF (1- \eta)^{d}\rho^{d- 2}\delta \rmd\boldsymbol\xi'\\ &\geqslant (1- \eta)^{d+ 1}\omega_d\rho^{d- 2}\delta,
\end{split}\end{equation}
Combining \eqref{volume formula}, \eqref{angular decomposition}, and \eqref{volume calculation}, and choosing $\eta$ small enough we arrive at \eqref{small interval statement}. The theorem is proved.
\enp

\bibliographystyle{amsplain}
\bibliography{IDSref}
\end{document}